\pdfoutput=1
\RequirePackage{ifpdf}
\ifpdf 
\documentclass[pdftex]{sigma}
\else
\documentclass{sigma}
\fi

\numberwithin{equation}{section}
\newtheorem{thm}{Theorem}[section]
\newtheorem{lem}[thm]{Lemma}
\newtheorem{cor}[thm]{Corollary}
\newtheorem{pro}[thm]{Proposition}
\theoremstyle{definition}{
\newtheorem{ex}[thm]{Example}
\newtheorem{rmk}[thm]{Remark}
\newtheorem{defi}[thm]{Definition}}

\newcommand{\lon}{\rightarrow}
\newcommand{\li}{\mathfrak l}
\newcommand{\g}{\frkg}
\newcommand{\huaO}{\mathcal{O}}
\newcommand{\frkg}{\mathfrak g}
\newcommand{\frkl}{\mathfrak l}
\newcommand{\dM}{\mathrm{d}}
\newcommand{\Hom}{\operatorname{Hom}}
\newcommand{\Ad}{\operatorname{Ad}}
\newcommand{\gl}{\mathfrak{gl}}
\newcommand{\ad}{\operatorname{ad}}

\begin{document}

\allowdisplaybreaks

\newcommand{\arXivNumber}{1507.04061}

\renewcommand{\PaperNumber}{014}

\FirstPageHeading

\ShortArticleName{Hom-Big Brackets: Theory and Applications}

\ArticleName{Hom-Big Brackets: Theory and Applications}

\Author{Liqiang CAI and Yunhe SHENG}
\AuthorNameForHeading{L.~Cai and Y.~Sheng}

\Address{Department of Mathematics, Jilin University, Changchun 130012, Jilin, China}
\Email{\href{mailto:cailq13@mails.jlu.edu.cn}{cailq13@mails.jlu.edu.cn}, \href{mailto:shengyh@jlu.edu.cn}{shengyh@jlu.edu.cn}}

\ArticleDates{Received July 16, 2015, in f\/inal form February 02, 2016; Published online February 05, 2016}

\Abstract{In this paper, we introduce the notion of hom-big brackets, which is a gene\-ra\-lization of Kosmann-Schwarzbach's big brackets. We show that it gives rise to a graded hom-Lie algebra. Thus, it is a useful tool to study hom-structures. In particular, we use it to describe hom-Lie bialgebras and hom-Nijenhuis operators.}


\Classification{17B70; 17B62}

\section{Introduction}

The notion of hom-Lie algebras was introduced by Hartwig, Larsson
and Silvestrov in \cite{HLS} as part of a study of deformations of
the Witt and the Virasoro algebras. In a hom-Lie algebra, the Jacobi
identity is twisted by a linear map, called the hom-Jacobi identity.
Some $q$-deformations of the Witt and the Virasoro algebras have the
structure of a hom-Lie algebra~\cite{HLS}. Because of their close relation
to discrete and deformed vector f\/ields and dif\/ferential calculus
\cite{HLS,LD1,LD2}, hom-Lie algebras were widely studied recently
\cite{AEM,MS2,MS1,Yao1,Yao2}.

The big bracket $\{\cdot,\cdot\}$ on $\wedge^\bullet(V\oplus V^*)$ is exactly the graded Poisson bracket on $T^*V[1]$. See \cite{Y1,KostantStenberg,LR} for more details. It was already clear that the big bracket was
the appropriate tool to study the theory of Lie bialgebras. Many generalizations are made for the big bracket and there are many applications, e.g., in the theory of strong homotopy bialgebras~\cite{olga}, in the theory of Poisson geometry and Lie algebroids \cite{Y3,KSpoisson_function,Y,Roytenberg}, in the theory of deformations of Courant algebroids \cite{costa,KSNijenhuis_CA}, and etc.

The purpose of this paper is to def\/ine the hom-analogue of the big bracket, i.e., the hom-big bracket, and provide a tool to study hom-structures. Since the Nijenhuis--Richardson brac\-ket~\cite{NR} on the direct sum $\oplus_k\Hom(\wedge^kV, V)$ is a part of the big bracket, f\/irst we def\/ine the hom-Nijenhuis--Richardson bracket $[\cdot,\cdot]_{\alpha}$, where $\alpha\in {\rm GL}(V)$, and show that the hom-Nijenhuis--Richardson bracket gives rise to a graded hom-Lie algebra. The hom-Nijenhuis--Richardson bracket has some good properties. On one hand, it can describe hom-Lie algebra structures, namely for $\mu\in\Hom(\wedge^2V,V)$, $[\mu,\mu]_\alpha=0$ if and only if $\mu$ satisf\/ies the hom-Jacobi identity. On the other hand, for $A,B\in\Hom(V,V)$, we have
\begin{gather*}
[A,B]_{\alpha}=\alpha A\alpha^{-1}B\alpha^{-1}-\alpha B\alpha^{-1}A\alpha^{-1}.
\end{gather*}
 This bracket \looseness=-1 is exactly the one introduced in~\cite{sheng2}, which plays an important role in the representation theory.
Then we introduce the hom-big bracket and show that it gives rise to a graded hom-Lie algebra. Moreover, it also gives rise to a purely hom-Poisson structure introduced in~\cite{LGT}.

As the f\/irst application, we def\/ine hom-Lie bialgebras using the hom-big bracket. A~Lie bialgebra~\cite{D}
 is the Lie-theoretic case of a~bialgebra: it is a set with a Lie algebra structure and a~Lie coalgebra
one which are compatible. Lie bialgebras are the inf\/initesimal
objects of Poisson--Lie groups. Both Lie bialgebras and Poisson--Lie
groups are considered as semiclassical limits of quantum groups. The
solutions of the classical Yang--Baxter equations provide examples of
Lie bialgebras. The hom-analogue of the Yang--Baxter equation and quantum groups are studied in~\cite{Yao1,Yao2}. Furthermore, hom-analogues of a Lie bialgebra are studied in two approaches recently~\cite{sheng1,Yao3}.
The hom-Lie bialgebra def\/ined here is the same as the one given in~\cite{Yao3}. As a~byproduct, we give the def\/initions of a hom-Lie quasi-bialgebra and a~hom-quasi-Lie bialgebra. We hope that they are connected with hom-quantum groups~\cite{Yao2}. They also provide a way to study the hom-analogue of Drinfeld twists.

As the second application, we def\/ine hom-Nijenhuis operators using the hom-big bracket. For a Lie algebra $(\li,[\cdot,\cdot]_\li)$, a Nijenhuis operator is
a linear map $N\colon \li\longrightarrow\li$ satisfying
\begin{gather*}
[Nx, Ny]_\li = N([Nx, y]_\li +[x, Ny]_\li - N[x,y]_\li),
\end{gather*}
which   gives a trivial deformation of Lie algebra $\li$ and plays an
important role in the study of integrability of Hamilton equations
\cite{Dorfman,Kosmann3}. In general, a 1-parameter inf\/initesimal
deformation is controlled by a 2-cocycle
$\omega\colon \wedge^2\li\longrightarrow\li$ (see \cite{NR} for more
details). In~\cite{grabowski}, the authors identif\/ied the role that
Nijenhuis operators play in the theory of contractions and
deformations of both Lie algebras and Leibniz (Loday) algebras. Nijenhuis operators on algebras other than Lie algebras, including
for $L_\infty$-algebras, Poisson structures and Courant algebroids, can be found in~\cite{costa, Nijenhuis_forms,CGM, G, KSNijenhuis_CA, Kosmann3}. In~\cite{sheng3}, a~notion of a hom-Nijenhuis operator was given. However, the hom-Nijenhuis operator def\/ined here is dif\/ferent from the existing one. Similarly, the notion of a hom-$\mathcal O$-operator is also dif\/ferent from the one given in~\cite{sheng1}. But we believe that the current def\/initions are more reasonable (see Remarks~\ref{rmk:Nijenhuis} and~\ref{rmk:O-operator}) and this justif\/ies the usage of the hom-big bracket.

 The paper is organized as follows. In Section~\ref{section2}, we recall notions of hom-Lie algebras, representations of hom-Lie algebras, hom-right-symmetric algebras, big brackets, Lie bialgebras and Nijenhuis operators. In Section~\ref{section3}, we give the def\/inition of the hom-Nijenhuis--Richardson bracket $[\cdot,\cdot]_{\alpha}$ and show that the composition gives rise to a hom-right-symmetric algebra structure (Theorem~\ref{thm:composition}). Consequently, $[\cdot,\cdot]_{\alpha}$ satisf\/ies the hom-Jacobi identity. Then we obtain a~new cohomology of a hom-Lie algebra via the hom-Nijenhuis--Richardson bracket, see~\eqref{eq:newd}. In Section~\ref{section4}, we give the def\/inition of the hom-big bracket
and show that it gives rise to a~graded hom-Lie algebra (Theorem~\ref{Graded}). In particular, it is consistent with the hom-Nijenhuis--Richardson bracket. In Section~\ref{section5}, we def\/ine a hom-Lie bialgebra using the hom-big bracket and describe it using the usual algebraic language. We also give the def\/initions of a hom-Lie quasi-bialgebra and a hom-quasi-Lie bialgebra. In Section~\ref{section6}, we def\/ine a hom-Nijenhuis operator and a hom-$\mathcal O$-operator using the hom-big bracket and study their properties.

\section{Preliminaries}\label{section2}

\subsection*{Hom-Lie algebras and hom-right-symmetric algebras}

\begin{defi}\quad
\begin{itemize}\item[(1)]
 A (multiplicative)   hom-Lie algebra is a triple $(V,[\cdot,\cdot],\alpha)$ consisting of a~vector space~$V$, a~skew-symmetric bilinear map (bracket) $[\cdot,\cdot]\colon \wedge^2V\longrightarrow
 V$ and a linear map $\alpha\colon V\lon V$ preserving the bracket, such that the following hom-Jacobi
 identity with respect to $\alpha$ is satisf\/ied
 \begin{gather*}
 [\alpha(x),[y,z]]+[\alpha(y),[z,x]]+[\alpha(z),[x,y]]=0.
 \end{gather*}

\item[(2)]A hom-Lie algebra is called a regular hom-Lie algebra if $\alpha$ is
an algebra automorphism.
 \end{itemize}
\end{defi}

\begin{defi}\label{repn}
 A representation of the hom-Lie algebra $(V,[\cdot,\cdot],\alpha)$ on
 the vector space $W$ with respect to $\beta\in\gl(W)$ is a linear map
 $\rho\colon V\longrightarrow \gl(W)$, such that for all
 $x,y\in V$, the following equalities are satisf\/ied
 \begin{gather*}
\rho(\alpha(x))\circ \beta = \beta\circ \rho(x),\\
 \rho([x,y])\circ
 \beta = \rho(\alpha(x))\circ\rho(y)-\rho(\alpha(y))\circ\rho(x).
 \end{gather*}
\end{defi}

\begin{defi}[\cite{MS2}]
A hom-right-symmetric algebra is a triple $(V,*,\gamma)$ consisting of a linear space $V$, a bilinear map $*\colon V\otimes V\rightarrow V$ and a linear map $\gamma\colon V\rightarrow V$ preserving the multiplication such that the following equality is satisf\/ied
\begin{gather*}
(x*y)*\gamma(z)-\gamma(x)*(y*z)=(x*z)*\gamma(y)-\gamma(x)*(z*y),\qquad\forall\, x, y, z\in V.
\end{gather*}
\end{defi}
Given a hom-right-symmetric algebra $(V,*,\gamma)$, def\/ine $[\cdot,\cdot]\colon \wedge^2V\longrightarrow V$ by $[x,y]=x*y-y*x$. Then $(V,[\cdot,\cdot],\gamma)$ is a hom-Lie algebra.

\subsection*{Big brackets, Lie bialgebras and Nijenhuis operators}

Let $V$ be a vector space, denote by
$\wedge^{n}(V\oplus V^{*})=\oplus_{p+q=n}(\wedge^{q+1}V^{*}\otimes \wedge^{p+1}V)$ and $\wedge^\bullet(V\oplus V^{*})=\oplus_{n=-2}^{\infty}\wedge^{n}(V\oplus V^{*})$.
We say that $u$ is of \emph{degree $|u|$}, if $u\in\wedge^{|u|}(V\oplus V^{*})$. The \emph{big bracket} $\{\cdot,\cdot\}\colon \wedge^k(V\oplus V^*)\otimes \wedge^l(V\oplus V^*)\longrightarrow \wedge^{k+l}(V\oplus V^*)$ is uniquely determined by the following properties:
\begin{itemize}\itemsep=0pt
 \item[(i)] for all $x,y\in V$, $\{x,y\}=0$;
 \item[(ii)] for all $\xi,\eta\in V^*$, $\{\xi,\eta\}=0$;
 \item[(iii)] for all $x\in V, ~\xi\in V^*$, $\{x,\xi\}=\xi(x)$;
 \item[(iv)] it is graded skew-symmetric, i.e., for all $e_1\in\wedge^k(V\oplus V^*)$, $e_2\in\wedge^l(V\oplus V^*)$, we have
 \begin{gather*}
 \{e_1,e_2\}=-(-1)^{kl}\{e_2,e_1\};
 \end{gather*}
 \item[(v)] for all $e\in\wedge^k(V\oplus V^*)$, $\{e,\cdot\}$ is a graded derivation, i.e., for all $e_1\in\wedge^l(V\oplus V^*)$ and $e_2\in\wedge^s(V\oplus V^*)$, we have
 \begin{gather*}
 \{e,e_1\wedge e_2\}=\{e,e_1\}\wedge e_2+(-1)^{kl}e_1\wedge\{e,e_2\}.
 \end{gather*}
\end{itemize}

A \emph{Lie bialgebra} is a triple $(V,\mu,\Delta)$, where $\mu\colon \wedge^2V\longrightarrow V$ and $\Delta\colon V\longrightarrow\wedge^2 V$ are linear maps (viewed as elements in $\wedge^2V^*\otimes V$ and $V^*\otimes \wedge^2 V$ respectively) such that\footnote{In the sequel, we use the same notation for a~multilinear map and the corresponding tensor.}
\begin{gather*}
\{\mu+\Delta,\mu+\Delta\}=0.
\end{gather*}
This is equivalent to
\begin{itemize}\itemsep=0pt
\item $(V,\mu)$ is a Lie algebra;
\item $(V^*,\Delta^*)$ is a Lie algebra;
\item $\Delta(\mu(x,y))=\ad^\mu_x\Delta(y)-\ad^\mu_y\Delta(x)$, where $\ad^\mu$ is the action of the Lie algebra $(V,\mu)$ on $\wedge^2V$ given by $\ad^\mu_x(y\wedge z)=\mu(x,y)\wedge z+y\wedge\mu(x,z)$.
\end{itemize}

Using the big bracket, a \emph{Nijenhuis operator} $N\colon V\longrightarrow V$ (viewed as an element in $V^*\otimes V$) on a Lie algebra $(V,\mu)$ can be described by
\begin{gather*}
\{N,\{N,\mu\} \} -\{N^2 ,\mu\}=0.
\end{gather*}

\section{The hom-Nijenhuis--Richardson bracket}\label{section3}

\looseness=-1
The Nijenhuis--Richardson bracket is a graded Lie algebra structure on the space of alternating multilinear forms of a vector space to itself, introduced by Nijenhuis and Richardson \cite{NR}. In this section, we introduce the notion of hom-Nijenhuis--Richardson brackets and study their properties.

Let $V$ be a vector space. For any $k\geq 0$, denote by $C^k(V, V)=\Hom(\wedge^kV, V)$ and $C(V,V)=\bigoplus^{\infty}_{k=0}C^k(V, V)$.
We say that $P\in C(V,V)$ is of \emph{degree $k$}, if $P\in C^{k+1}(V,V)$.

\begin{defi}
Let $\alpha\in {\rm GL}(V)$ be an invertible linear map. The hom-Nijenhuis--Richardson bracket
\begin{gather*}
[ \cdot, \cdot ]_{\alpha} \colon \ C^{k+1}(V, V) \times C^{l+1}(V, V) \rightarrow C^{k+l+1}(V, V)
\end{gather*}
is def\/ined by
\begin{gather}\label{hom-NR bracket}
[P, Q]_{\alpha} = P\circ Q- (-1)^{kl} Q\circ P, \qquad\forall \, P\in C^{k+1}(V, V), \quad Q\in C^{l+1}(V, V),
\end{gather}
where the composition $\circ$ is given by
\begin{gather}
 (P\circ Q) (x_1,\dots,x_{k+l+1})=\sum_{\sigma \in (l+1, k)\text{-unshuf\/f\/les}} \operatorname{sgn}(\sigma )\alpha P\big( \alpha^{-1} Q\big( \alpha^{-1} x_{\sigma (1) }, \dots, \alpha^{-1} x_{\sigma (l+1)}\big) ,\nonumber\\
 \hphantom{(P\circ Q) (x_1,\dots,x_{k+l+1})=\sum_{\sigma \in (l+1, k)\text{-unshuf\/f\/les}}}{}
 \alpha^{-1} x_{\sigma (l+2) }, \dots, \alpha^{-1} x_{\sigma(k+l+1)}\big), \label{hom-NR brackethalf}
\end{gather}
in which $\alpha^{-1}$ is the inverse of $\alpha$.
\end{defi}

\begin{ex} For all $N\in C^{1}(V,V)$ and $y\in C^{0}(V,V)$, we have
\begin{gather*}
[N,y]_{\alpha}=(\Ad_{\alpha}N)(y),
\end{gather*}
where $\Ad_\alpha$ is the adjoint map, i.e., $\Ad_\alpha N=\alpha N\alpha^{-1}$.
\end{ex}

\begin{ex}{\rm For all $P,~Q \in C^1(V, V)=\frkg\frkl(V)$, we have
\begin{gather*}
[P, Q]_{\alpha}= \alpha P\alpha^{-1}Q\alpha^{-1}- \alpha Q\alpha^{-1}P\alpha^{-1}.
\end{gather*}
In \cite{sheng2}, the authors showed that $(\gl(V),[\cdot,\cdot]_\alpha,\Ad_\alpha)$ is a hom-Lie algebra, which plays important roles in the representation theory of hom-Lie algebras. More precisely, any representation of a hom-Lie algebra $\g$ on $V$ can be realized as a homomorphism from~$\g$ to the hom-Lie algebra $(\gl(V),[\cdot,\cdot]_\alpha,\Ad_\alpha)$.}
 \end{ex}

The Jacobi identity can be described by the Nijenhuis--Richardson bracket. Similarly, the hom-Jacobi identity can be described by
the hom-Nijenhuis--Richardson bracket.

\begin{lem}\label{lem:mu}
Let $\mu\in C^2(V,V)$ and $\alpha\in {\rm GL}(V)$. Then $(V,\mu,\alpha)$ is a hom-Lie algebra if and only if $\Ad_\alpha\mu=\mu$ and
\begin{gather*}
[\mu,\mu]_{\alpha}=0.
\end{gather*}
\end{lem}
\begin{proof}
It is straightforward to see that $\Ad_\alpha\mu=\mu$ is equivalent to that
\begin{gather}\label{eq:mor}
 \mu(\alpha(x),\alpha(y))=\alpha\mu(x,y).
\end{gather}
By \eqref{eq:mor}, \eqref{hom-NR bracket} and \eqref{hom-NR brackethalf}, we have
\begin{gather*}
[\mu,\mu]_{\alpha}(\alpha x_1,\alpha x_2.\alpha x_3)=2\big(\mu(\mu(x_1,x_2),\alpha x_3)+\mu(\mu(x_2,x_3),\alpha x_1)+\mu(\mu(x_3,x_1),\alpha x_2)\big).
\end{gather*}
Thus, $[\mu,\mu]_{\alpha}=0$ is equivalent to the hom-Jacobi identity.
\end{proof}

About the properties of the composition~$\circ$ given by~\eqref{hom-NR brackethalf}, we have
\begin{thm}\label{thm:composition}
With the above notations, $(C(V,V),\circ,\Ad_{\alpha})$ is a hom-right-symmetric algebra, i.e., for all $P,Q,W\in C(V,V)$, the following equalities hold
\begin{gather}
\label{eq:right1}\Ad_{\alpha}(P\circ Q) = \Ad_{\alpha}P\circ\Ad_{\alpha}Q, \\
\label{eq:right2}(P\circ Q)\circ \Ad_{\alpha}W-\Ad_{\alpha}P\circ(Q\circ W) = (P\circ W)\circ \Ad_{\alpha}Q-\Ad_{\alpha}P\circ (W\circ Q),
\end{gather}
where $\Ad_\alpha\colon C^{k+1}(V, V)\longrightarrow C^{k+1}(V, V)$ is given by
\begin{gather}\label{eq:Ad}
 \Ad_\alpha P(x_1,\dots, x_{k+1})=\alpha P\big(\alpha^{-1}(x_1),\dots,\alpha^{-1}(x_{k+1})\big).
\end{gather}
\end{thm}

\begin{proof}
 For all $P\in C^{k+1}(V,V)$, $Q\in C^{l+1}(V,V)$, $W\in C^{m+1}(V,V)$ and $x_{1},\dots,x_{k+l+m+1}\in V$, by \eqref{hom-NR brackethalf}, we have
\begin{gather*}
 \Ad_{\alpha}(P\circ Q)(x_1, \dots , x_{k+l+1})\\
 \qquad{}
 = \sum_{\sigma \in (l+1, k)\text{-unshuf\/f\/les} } \operatorname{sgn}(\sigma )\alpha^{2} P\big( \alpha^{-1} Q\big( \alpha^{-2} x_{\sigma (1) }, \dots, \alpha^{-2} x_{\sigma (l+1)}\big) ,\\
 \hphantom{qquad{} = \sum_{\sigma \in (l+1, k)\text{-unshuf\/f\/les} }}{}
 \alpha^{-2} x_{\sigma (l+2) }, \dots, \alpha^{-2} x_{\sigma(k+l+1)}\big)
 \\
\qquad{}
 = \sum_{\sigma \in (l+1, k)\text{-unshuf\/f\/les} } \operatorname{sgn}(\sigma )
 \alpha\Ad_{\alpha}(P)\big( \alpha^{-1}\Ad_{\alpha}(Q)\big( \alpha^{-1} x_{\sigma (1) }, \dots, \alpha^{-1} x_{\sigma (l+1)}\big) , \\
\hphantom{\qquad{}= \sum_{\sigma \in (l+1, k)\text{-unshuf\/f\/les} }}{}
 \alpha^{-1} x_{\sigma (l+2) }, \dots, \alpha^{-1} x_{\sigma(k+l+1)}\big) \\
\qquad{}
=(\Ad_{\alpha}(P)\circ\Ad_{\alpha}(Q))(x_1, \dots , x_{k+l+1}),
\end{gather*}
which implies that \eqref{eq:right1} holds.

Moreover, we have
\begin{gather*}
 ( (P\circ Q)\circ \Ad_{\alpha}W-\Ad_{\alpha}P\circ(Q\circ W)-(P\circ W)\circ \Ad_{\alpha}Q\\
\qquad\quad {}+\Ad_{\alpha}P\circ (W\circ Q) )(x_{1},\dots,x_{k+l+m+1}) \\
\qquad{}=\sum_{\sigma\in(m+1,k+l)\text{-unshuf\/f\/les}}\operatorname{sgn}(\sigma)
\alpha(P\circ Q)\big(W\big(\alpha^{-2}x_{\sigma(1)},\dots,\alpha^{-2}x_{\alpha(m+1)}\big),\\
\hphantom{\qquad{}=\sum_{\sigma\in(m+1,k+l)\text{-unshuf\/f\/les}}}{}
\alpha^{-1}x_{\sigma(m+2)},\dots,\alpha^{-1}x_{\sigma(k+l+m+1)}\big)
\\
\qquad\quad{}
-\sum_{\sigma\in(l+m+1,k)\text{-unshuf\/f\/les}}\operatorname{sgn}(\sigma)
\alpha^{2}P\big(\alpha^{-2}(Q\circ W)\big(\alpha^{-1}x_{\sigma(1)},\dots,\alpha^{-1}x_{\sigma(l+m+1)}\big),\\
\hphantom{\qquad{}-\sum_{\sigma\in(l+m+1,k)\text{-unshuf\/f\/les}}}{}
\alpha^{-2}x_{\sigma(l+m+2)},\dots,\alpha^{-2}x_{\sigma(k+l+m+1)}\big) \\
\qquad\quad{}
-\sum_{\sigma\in(l+1,k+m)\text{-unshuf\/f\/les}}\operatorname{sgn}(\sigma)
\alpha(P\circ W)\big(Q\big(\alpha^{-2}x_{\sigma(1)},\dots,\alpha^{-2}x_{\sigma(l+1)}\big),\\
\hphantom{\qquad{}=\sum_{\sigma\in(m+1,k+l)\text{-unshuf\/f\/les}}}{}
\alpha^{-1}x_{\sigma(l+2)},\dots,\alpha^{-1}x_{\sigma(k+l+m+1)}\big)\\
\qquad\quad{}
+\sum_{\sigma\in(l+m+1,k)\text{-unshuf\/f\/les}}\operatorname{sgn}(\sigma)
\alpha^{2}P\big(\alpha^{-2}(W\circ Q)\big(\alpha^{-1}x_{\sigma(1)},\dots,\alpha^{-1}x_{\sigma(l+m+1)}\big),\\
\hphantom{\qquad{}=\sum_{\sigma\in(m+1,k+l)\text{-unshuf\/f\/les}}}{}
\alpha^{-2}x_{\sigma(l+m+2)},\dots,\alpha^{-2}x_{\sigma(k+l+m+1)}\big)\\
\qquad{}
=\sum_{\sigma\in(m+1,k+l)\text{-unshuf\/f\/les}}\sum_{\tau\in(l+1,k)\text{-unshuf\/f\/les}}(-1)^{l+1}
\operatorname{sgn}(\sigma)\operatorname{sgn}(\tau) \\
\qquad\quad{}
\alpha^{2}P\big(\alpha^{-1}Q\big(\alpha^{-2}x_{\tau\sigma(m+2)},\dots,\alpha^{-2}x_{\tau\sigma(m+l+2)}\big),
\alpha^{-1}W\big(\alpha^{-2}x_{\sigma(1)},\dots,\alpha^{-2}x_{\sigma(m+1)} \big), \\
\hphantom{\qquad{}=\sum_{\sigma\in(m+1,k+l)\text{-unshuf\/f\/les}}}{}
\alpha^{-2}x_{\tau\sigma(m+l+3)},\dots,\alpha^{-2}x_{\tau\sigma(k+l+m+1)}\big) \\
\qquad\quad{}
-\sum_{\sigma\in(l+1,k+m)\text{-unshuf\/f\/les}}\sum_{\tau\in(m+1,k)\text{-unshuf\/f\/les}}(-1)^{l+1}\operatorname{sgn}(\sigma)\operatorname{sgn}(\tau) \\
\qquad\quad{}
\alpha^{2}P\big(\alpha^{-1}W\big(\alpha^{-2}x_{\tau\sigma(l+2)},\dots,\alpha^{-2}x_{\tau\sigma(m+l+2)}\big),
\alpha^{-1}Q\big(\alpha^{-2}x_{\sigma(1)},\dots,\alpha^{-2}x_{\sigma(l+1)}\big),\nonumber\\
\hphantom{\qquad{}=\sum_{\sigma\in(m+1,k+l)\text{-unshuf\/f\/les}}}{}
\alpha^{-2}x_{\tau\sigma(m+l+3)},\dots,\alpha^{-2}x_{\tau\sigma(k+l+m+1)}\big)
=0.
\end{gather*}
Thus, \eqref{eq:right2} holds.
Therefore, $(C(V,V),\circ,\Ad_{\alpha})$ is a hom-right-symmetric algebra.
\end{proof}

\begin{cor}\label{cor:hom-Lie}
$(C(V,V),[\cdot,\cdot]_{\alpha}, \Ad_{\alpha})$ is a graded hom-Lie algebra, i.e., we have
\begin{gather}
\label{graded hom-Lie alg1}
\Ad_{\alpha}([P, Q]_{\alpha}) = [\Ad_{\alpha}(P), \Ad_{\alpha}(Q)]_{\alpha},\\
[\Ad_{\alpha}W, [P, Q]_{\alpha}]_{\alpha}
 = [[W, P]_{\alpha}, \Ad_{\alpha}Q]_{\alpha}+(-1)^{|W|\cdot|P|}[\Ad_{\alpha}P, [W, Q]_{\alpha}]_{\alpha}.\nonumber
\end{gather}
\end{cor}

By Lemma \ref{lem:mu} and Corollary~\ref{cor:hom-Lie}, for any hom-Lie algebra $(V,\mu,\alpha)$, there is a coboundary operator $\dM\colon C^k(V,V)\longrightarrow
C^{k+1}(V,V)$, which is given by
\begin{gather}\label{eq:newd}
\dM f=(-1)^{k+1}[\mu,f]_{\alpha},\qquad \forall\, f\in C^{k}(V,V).
\end{gather}

This formula can be easily generalized for any representation. More precisely, for any representation $\rho$ of the hom-Lie algebra $(V,\mu,\alpha)$ on $W$ with respect to $\beta$, def\/ine
 $\dM\colon \Hom(\wedge^kV,W)\longrightarrow
\Hom(\wedge^{k+1}V,W)$ by
\begin{gather*}
\dM f(x_1,\dots,x_{k+1})=\sum_{i=1}^{k+1}(-1)^{i+1}\rho(x_i)\big(f\big(\alpha^{-1}x_1,\dots,\widehat{\alpha^{-1}x_i},\dots,\alpha^{-1}x_{k+1}\big)\big) \\
\qquad{}+\sum_{i<j}(-1)^{i+j}\beta f\big(\mu(\alpha^{-2}x_i,\alpha^{-2}x_j),\alpha^{-1}x_1,\dots,\widehat{\alpha^{-1}x_i},\dots,\widehat{\alpha^{-1}x_j},\dots,\alpha^{-1}x_{k+1}\big).
\end{gather*}

\begin{thm} With above notations,
$\dM^{2}=0$. Thus, we have a well-defined cohomology.
\end{thm}
\begin{proof}
By straightforward computation.
\end{proof}

\begin{rmk}
The coboundary operator $\dM$ given above is dif\/ferent from the one given in~\cite{sheng2}. It turns out that cohomology theories are not unique for hom-Lie algebras.
\end{rmk}

\section{The hom-big bracket}\label{section4}

In this section, we introduce the notion of hom-big brackets.
Let $\alpha\colon V\rightarrow V$ be an invertible linear map, $\alpha^{-1}$ its inverse and $\alpha^{*} \colon V^{*} \rightarrow V^{*}$ its dual map.
 $\alpha$ induces a linear map from $\wedge^{p+1}V$ to $\wedge^{p+1}V$, for which we use the same notation, by
 \begin{gather*}
 \alpha(X)=\alpha(x_{0})\wedge\cdots\wedge\alpha(x_{p}),\qquad \forall\, X=x_{0}\wedge\cdots\wedge x_{p}\in\wedge^{p+1}V.
\end{gather*}
 Similarly, $(\alpha^{-1})^{*}\colon \wedge^{q+1}V^{*}\longrightarrow \wedge^{q+1}V^{*} $ is given by
\begin{gather*}
 \big(\alpha^{-1}\big)^{*}(\Xi)=\big(\alpha^{-1}\big)^{*}(\xi_{0})\wedge\cdots\wedge\big(\alpha^{-1}\big)^{*}(\xi_{q}),\qquad \forall \, \Xi=\xi_{0}\wedge\cdots\wedge\xi_{q}\in\wedge^{q+1}V^{*}.
 \end{gather*}
Def\/ine $\Ad_\alpha\colon \wedge^{q+1}V^{*}\otimes \wedge^{p+1}V\longrightarrow \wedge^{q+1}V^{*}\otimes \wedge^{p+1}V$ by
\begin{gather}\label{eq:Add}
\Ad_{\alpha}(\Xi\otimes X)=\big(\alpha^{-1}\big)^{*}(\Xi)\otimes\alpha(X).
\end{gather}
In particular, we have
\begin{gather*}
\Ad_{\alpha}(\Xi)=\big(\alpha^{-1}\big)^{*}(\Xi),\qquad
\Ad_{\alpha}(X)=\alpha(X).
\end{gather*}
\begin{rmk}
 For $X\in V$, \eqref{eq:Add} is consistent with \eqref{eq:Ad}. Thus, we use the same notation.
\end{rmk}

\begin{defi}
For an invertible linear map $\alpha\in {\rm GL}(V)$, on the graded vector space $\wedge^\bullet(V\oplus V^{*})$, we def\/ine the hom-big bracket:
\begin{gather*}
\{\cdot, \cdot \}_{\alpha}\colon \ \wedge^{n}(V\oplus V^{*})\otimes\wedge^{m}(V\oplus V^{*})\rightarrow \wedge^{n+m}(V\oplus V^{*}),
\end{gather*}
which is uniquely determined by the following properties:
\begin{itemize}\item[(i)]
for all $x, y\in V$,
$\{x, y\}_{\alpha}=0$;

\item[(ii)]
for all $\xi, \eta\in V^{*}$,
$\{\xi, \eta\}_{\alpha}=0$;

\item[(iii)]
For all $x\in V$ and $\xi\in V^{*}$,
$\{x, \xi\}_{\alpha}=\xi(\alpha^{-1}x)$;
\item[(iv)] $\{\cdot,\cdot\}_{\alpha}$ satisf\/ies the following graded-commutative relation:
\begin{gather}\label{comm}
\{u, v\}_{\alpha}=-(-1)^{|u|\cdot |v|}\{v, u\}_{\alpha};
\end{gather}

\item[(v)]
$\{~,\cdot\}_{\alpha}$ is a graded $\Ad_{\alpha}$-derivation, i.e., for all $u,v,w\in\wedge^\bullet(V\oplus V^*)$, we have
\begin{gather}\label{lb}
\{u\wedge v, w\}_{\alpha}=\Ad_{\alpha}(u)\wedge\{v, w\}_{\alpha}+(-1)^{|v|\cdot|w|}\{u,w\}_{\alpha}\wedge\Ad_{\alpha}(v).
\end{gather}
\end{itemize}
\end{defi}

In the classical case, the big bracket gives rise to a graded Lie algebra structure. Similarly, the hom-big bracket also induces a graded hom-Lie algebra structure, which is the main result in this section.

\begin{thm}\label{Graded}
With the above notations, $(\wedge^\bullet(V\oplus V^{*}),\{\cdot,\cdot\}_{\alpha},\Ad_{\alpha})$ is a graded hom-Lie algebra, i.e.,
for all $P,Q,W\in\wedge^\bullet(V\oplus V^{*})$, we have
\begin{gather}\label{G1}
\Ad_{\alpha}\{P, Q\}_{\alpha}=\{\Ad_{\alpha}(P), \Ad_{\alpha}(Q)\}_{\alpha},\\
\label{G2}
\{\Ad_{\alpha}P, \{Q, W\}_{\alpha}\}_{\alpha}=\{\{P, Q\}_{\alpha}, \Ad_{\alpha}W\}_{\alpha}+(-1)^{|P|\cdot|Q|}\{\Ad_{\alpha}Q, \{P, W\}_{\alpha}\}_{\alpha}.
\end{gather}
\end{thm}

\begin{rmk}
By Theorem~\ref{Graded} and the Leibniz rule~\eqref{lb}, $(\wedge^\bullet(V\oplus V^{*}),\wedge,\{\cdot,\cdot\}_{\alpha},\Ad_{\alpha})$ is a~purely hom-Poisson algebra, a notion which was introduced by Laurent-Gengoux and Teles in~\cite{LGT}.
\end{rmk}

To prove the theorem, we need some preparations.
For all $\Xi\in\wedge^{q+1}V^{*}$, $x,x_{1},\dots,x_{q}\in V$, def\/ine the \emph{interior product}
$i_{x}^{\alpha}\colon \wedge^{q+1}V^{*}\rightarrow\wedge^{q}V^{*}$
by
\begin{gather*}
(i_{x}^{\alpha}\Xi)(x_{1},\dots,x_{q}):=\Xi\big(\alpha^{-1}x,\alpha^{-1}x_{1},\dots,\alpha^{-1}x_{q}\big).
\end{gather*}

We can get the following formulas by straightforward computations.
\begin{lem}\label{interior}
For all $\Xi\in\wedge^{\bullet}V^{*}$, $x,y,z\in V$, we have
\begin{gather*}
\big(\alpha^{-1}\big)^*(i_{x}^{\alpha}\Xi) = i_{\alpha(x)}^{\alpha}\big(\big(\alpha^{-1}\big)^*(\Xi)\big),\qquad
i_{\alpha(y)}^{\alpha}i_{z}^{\alpha}\Xi = -i_{\alpha(z)}^{\alpha}i_{y}^{\alpha}\Xi.
\end{gather*}
\end{lem}

Furthermore, we have
\begin{lem}
For all $x\in V$, $\Xi=\xi_{0}\wedge\cdots \wedge\xi_{q}\in\wedge^{q+1}V^{*}$, $\Pi\in\wedge^{l+1}V^{*},~X=x_{0}\wedge\cdots \wedge x_{p}\in\wedge^{p+1}V$, $Y=y_{0}\wedge\cdots \wedge y_{k}\in\wedge^{k+1}V$, we have
\begin{gather}\label{formula1}
\{x, \Xi\}_{\alpha} = i_{x}^{\alpha}\Xi, \\
\label{formula2}
\{X, \Pi\}_{\alpha} = (-1)^{pl+p}\sum_{j=0}^{p}(-1)^{j}i_{x_{j}}^{\alpha}\Pi\otimes\alpha(x_{0}\wedge\cdots\wedge\widehat{x_{j}}\wedge\cdots\wedge x_{p}),\\
\{\Xi\otimes X, \Pi\otimes Y\}_{\alpha} = (-1)^{pl+p}\sum_{j=0}^{p}(-1)^{j}\big(\alpha^{-1}\big)^*(\Xi)\wedge i_{x_{j}}^{\alpha}\Pi\nonumber\\
\hphantom{\{X, \Pi\}_{\alpha} =}{}\qquad{}
\otimes\alpha(x_{0}\wedge\cdots\wedge\widehat{x_{j}}\wedge\cdots\wedge x_{p})\wedge\alpha(Y)\label{formula3}\\
\hphantom{\{X, \Pi\}_{\alpha} =}{} -(-1)^{pl+q}\sum_{j=0}^{k}(-1)^{j}i_{y_{j}}^{\alpha}\Xi\wedge\big(\alpha^{-1}\big)^*(\Pi)
\otimes\alpha(X)\wedge\alpha(y_{0}\wedge\cdots\wedge\widehat{y_{j}}\wedge\cdots\wedge y_{k}).\nonumber
\end{gather}
\end{lem}
\begin{proof}
Since $\{~,\cdot\}_{\alpha}$ is an $\Ad_{\alpha}$-derivation, we have
\begin{gather*}
\{x, \Xi\}_{\alpha}
 = \{x,\xi_{0}\wedge\cdots\wedge\xi_{q-1}\wedge\xi_{q}\}_{\alpha} \\
 \hphantom{\{x, \Xi\}_{\alpha}}{}
 = \{x,\xi_{0}\wedge\cdots\wedge\xi_{q-1}\}_{\alpha}\wedge\big(\alpha^{-1}\big)^*(\xi_{q})
+(-1)^{q}\big(\alpha^{-1}\big)^*(\xi_{0}\wedge\cdots\wedge\xi_{q-1})\wedge\{x,\xi_{q}\}_{\alpha} \\
\hphantom{\{x, \Xi\}_{\alpha}}{}
 = \sum_{j=0}^{q}(-1)^{j}\big(\alpha^{-1}\big)^*(\xi_{0}\wedge\cdots\wedge\xi_{j-1})\wedge\{x,\xi_{j}\}_{\alpha}\wedge\big(\alpha^{-1}\big)^*(\xi_{j+1}
 \wedge\cdots\wedge\xi_{q})
 = i_{x}^{\alpha}\Xi.
\end{gather*}
Similarly, we have
\begin{gather*}
\{X, \Pi\}_{\alpha}
 = \{x_{0}\wedge x_{1}\wedge\cdots\wedge x_{p},\Pi\}_{\alpha} \\
 \hphantom{\{X, \Pi\}_{\alpha}}{}
 = \alpha(x_{0})\wedge\{x_{1}\wedge\cdots\wedge x_{p},\Pi\}_{\alpha}+(-1)^{p(l-1)}\{x_{0},\Pi\}_{\alpha}\otimes\alpha(x_{1}\wedge\cdots\wedge x_{p}) \\
\hphantom{\{X, \Pi\}_{\alpha}}{}
 = (-1)^{pl+p}\sum_{j=0}^{p}(-1)^{j}i_{x_{j}}^{\alpha}\Pi\otimes\alpha(x_{0}\wedge\cdots\wedge\hat{x_{j}}\wedge\cdots\wedge x_{p}).
\end{gather*}
By \eqref{comm}, \eqref{lb}, \eqref{formula1} and \eqref{formula2}, we have
\begin{gather*}
 \{\Xi\otimes X, \Pi\otimes Y\}_{\alpha}
 = \{\Xi\otimes X,\Pi\}_{\alpha}\wedge\alpha(Y)+(-1)^{|\Xi\otimes X|\cdot |\Pi|}\big(\alpha^{-1}\big)^*(\Pi)\wedge\{\Xi\otimes X,Y\}_{\alpha} \\
 \qquad{}
 = \big(\alpha^{-1}\big)^*(\Xi)\wedge\{X,\Pi\}_{\alpha}\wedge\alpha(Y)+(-1)^{|\Xi\otimes X|\cdot|\Pi|+|X|\cdot|Y|}\big(\alpha^{-1}\big)^*(\Pi)\wedge\{\Xi,Y\}_{\alpha}\wedge\alpha(X) \\
\qquad{}
 = \big(\alpha^{-1}\big)^*(\Xi)\wedge\left( (-1)^{pl+p}\sum_{j=0}^{p}(-1)^{j}i_{x_{j}}^{\alpha}\Pi\otimes\alpha(x_{0}\wedge\cdots\wedge\widehat{x_{j}}\wedge\cdots\wedge x_{p})\right) \wedge\alpha(Y) \\
\qquad\quad{}
 +(-1)^{(p+q)l+(p-1)k+1}\big(\alpha^{-1}\big)^*(\Pi)\\
\qquad\quad{}
 \wedge\left( \sum_{j=0}^{k}(-1)^{j} i_{y_{j}}^{\alpha}\Xi\otimes\alpha(y_{0}\wedge\cdots\wedge\widehat{y_{j}}\wedge\cdots\wedge y_{k})\right) \wedge\alpha(X) \\
\qquad{}
 = (-1)^{pl+p}\sum_{j=0}^{p}(-1)^{j}\big(\alpha^{-1}\big)^*(\Xi)\wedge i_{x_{j}}^{\alpha}\Pi\otimes\alpha(x_{0}\wedge\cdots\wedge\widehat{x_{j}}\wedge\cdots\wedge x_{p})\wedge\alpha(Y) \\
\qquad\quad{}
 -(-1)^{pl+q}\sum_{j=0}^{k}(-1)^{j}i_{y_{j}}^{\alpha}\Xi\wedge\big(\alpha^{-1}\big)^*(\Pi)\otimes\alpha(X)\wedge\alpha(y_{0}\wedge\cdots\wedge\widehat{y_{j}}\wedge\cdots\wedge y_{k}). \tag*{\qed}
\end{gather*}
\renewcommand{\qed}{}
\end{proof}

\begin{proof}[The proof of Theorem~\ref{Graded}]
We just need to prove the case: $P=\Xi\otimes X$, $Q=\Pi\otimes Y$, $W=\Theta\otimes Z$.
By~\eqref{formula3}, we have
\begin{gather*}
 \Ad_{\alpha}\{P,Q\}_{\alpha} \\
 = \Ad_{\alpha} \bigg( (-1)^{pl+p}\sum_{j=0}^{p}(-1)^{j}\big(\alpha^{-1}\big)^*(\Xi)\wedge i_{x_{j}}^{\alpha}\Pi\otimes\alpha(x_{0}\wedge\cdots\wedge\hat{x_{j}}\wedge\cdots\wedge x_{p})\wedge\alpha(Y) \\
 \quad{}
 -(-1)^{pl+q}\sum_{j=0}^{k}(-1)^{j}i_{y_{j}}^{\alpha}\Xi\wedge\big(\alpha^{-1}\big)^*(\Pi)\otimes\alpha(X)\wedge\alpha(y_{0}\wedge\cdots\wedge\widehat{y_{j}}\wedge\cdots\wedge y_{k})\bigg ) \\
 = (-1)^{pl+p}\sum_{j=0}^{p}(-1)^{j}\big(\alpha^{-2}\big)^*(\Xi)\wedge\big(\alpha^{-1}\big)^*(i_{x_{j}}^{\alpha}\Pi)\otimes\alpha^2(x_{0}\wedge\cdots\wedge\widehat{x_{j}}\wedge\cdots\wedge x_{p})\wedge\alpha^2(Y) \\
\quad{} -(-1)^{pl+q}\sum_{j=0}^{k}(-1)^{j}\big(\alpha^{-1}\big)^*(i_{y_{j}}^{\alpha}\Xi)\wedge\big(\alpha^{-2}\big)^*(\Pi)\otimes\alpha^2(X)\wedge\alpha^2(y_{0}\wedge\cdots\wedge\widehat{y_{j}}\wedge\cdots\wedge y_{k}).
\end{gather*}
On the other hand, we have
\begin{gather*}
 \{\Ad_{\alpha}P,\Ad_{\alpha}Q\}_{\alpha}
 = \big\{\big(\alpha^{-1}\big)^*(\Xi)\otimes\alpha(X),\big(\alpha^{-1}\big)^*(\Pi)\otimes\alpha(Y)\big\}_{\alpha} \\
 = (-1)^{pl+p}\sum_{j=0}^{p}(-1)^{j}\big(\alpha^{-2}\big)^*(\Xi)\wedge i_{\alpha(x_{j})}^{\alpha}\big(\big(\alpha^{-1}\big)^*(\Pi)\big)\otimes\alpha^2(x_{0}\wedge\cdots\wedge\widehat{x_{j}}\wedge\cdots\wedge x_{p})\wedge\alpha^2(Y) \\
 \quad{}
 -(-1)^{pl+q}\sum_{j=0}^{k}(-1)^{j}i_{\alpha(y_{j})}^{\alpha}\big(\big(\alpha^{-1}\big)^*(\Xi)\big)
 \wedge\big(\alpha^{-2}\big)^*(\Pi)\otimes\alpha^2(X)\\
 \quad{} \wedge\alpha^2(y_{0}\wedge\cdots\wedge\widehat{y_{j}}\wedge\cdots\wedge y_{k}).
\end{gather*}
By Lemma~\ref{interior}, we get \eqref{G1}.

By the Leibniz rule, \eqref{G2} is equivalent to
\begin{gather}
\{\{\Xi,Y\}_{\alpha},\alpha(Z)\}_{\alpha} = -(-1)^{|\Xi|\cdot|Y|}\{\alpha(Y),\{\Xi,Z\}_{\alpha}\}_{\alpha},\nonumber\\
\label{Adjacobi}
\big\{\big(\alpha^{-1}\big)^*(\Xi),\{\Pi,Z\}_{\alpha}\big\}_{\alpha} = (-1)^{|\Xi|\cdot|\Pi|}
\big\{\big(\alpha^{-1}\big)^*(\Pi),\{\Xi,Z\}_{\alpha}\big\}_{\alpha}.
\end{gather}
By straightforward computations, we have
\begin{gather*}
\{\{\Xi,Y\}_{\alpha},\alpha(Z)\}_{\alpha} = (-1)^{k+1}\sum_{i=0}^{k}\sum_{j=0}^{m}(-1)^{i+j}i_{\alpha(z_{j})}^{\alpha}i_{y_{i}}^{\alpha}\Xi
\otimes\alpha^2(y_{0}\wedge\cdots\wedge\widehat{y_{i}}\wedge\cdots\wedge y_{k}) \\
\hphantom{\{\{\Xi,Y\}_{\alpha},\alpha(Z)\}_{\alpha} =}{}
 \wedge\alpha^2(z_{0}\wedge\cdots\wedge\widehat{z_{j}}\wedge\cdots\wedge z_{m}).
\end{gather*}
Similarly, we have
\begin{gather*}
\{\alpha(Y),\{\Xi,Z\}_{\alpha}\}_{\alpha} = -(-1)^{q+kq}\sum_{j=0}^{m}\sum_{i=0}^{k}(-1)^{j+i}
 i_{\alpha(y_{i})}^{\alpha}i_{z_{j}}^{\alpha}\Xi\\
 \hphantom{\{\alpha(Y),\{\Xi,Z\}_{\alpha}\}_{\alpha} =}{}
 \otimes\alpha^2(y_{0}\wedge\cdots\wedge\widehat{y_{i}}\wedge\cdots\wedge y_{k})\wedge\alpha^2(z_{0}\wedge\cdots\wedge\widehat{z_{j}}\wedge\cdots\wedge z_{m}).
\end{gather*}
By Lemma~\ref{interior}, we get $\{\{\Xi,Y\}_{\alpha},\alpha(Z)\}_{\alpha}=-(-1)^{|\Xi|\cdot|Y|}\{\alpha(Y),\{\Xi,Z\}_{\alpha}\}_{\alpha}$.

Similarly, we can prove that \eqref{Adjacobi} holds.
\end{proof}

At the end of this section, we show that the hom-big bracket is consistent with the hom-Nijenhuis--Richardson bracket.
\begin{pro}\label{pro:consistant}
For all $\Xi\in\wedge^{q+1}V^{*}$, $\Pi\in\wedge^{l+1}V^{*}$, $x,y\in V$, we have
\begin{gather}\label{relation}
\{\Xi\otimes x, \Pi\otimes y\}_{\alpha}=-(-1)^{ql} [\Xi\otimes x, \Pi\otimes y]_{\alpha}.
\end{gather}
\end{pro}

\begin{proof}
By \eqref{formula3}, we have
\begin{gather*}
\{\Xi\otimes x,\Pi\otimes y\}_{\alpha}(x_{1},\dots,x_{q+l+1})\nonumber\\
 =\big(\big(\alpha^{-1}\big)^*(\Xi)\wedge i_{x}^{\alpha}\Pi\otimes\alpha(y)-(-1)^{ql}\big(\alpha^{-1}\big)^*(\Pi)\wedge i_{y}^{\alpha}\Xi\otimes\alpha(x)\big)(x_{1},\dots,x_{q+l+1})\nonumber\\
=\frac{1}{(q+1)!l!}\sum_{\sigma\in S_{q+l+1}}\operatorname{sgn}(\sigma)
\big(\alpha^{-1}\big)^*(\Xi)(x_{\sigma(1)},\dots,x_{\sigma(q+1)})(i_{x}^{\alpha}\Pi)(x_{\sigma(q+2)},\dots,x_{\sigma(q+l+1)})\alpha(y) \\
\quad{}
-(-1)^{ql}\frac{1}{(l+1)!q!}\sum_{\sigma\in S_{q+l+1}}\operatorname{sgn}(\sigma) \\
\quad{}\times
\big(\alpha^{-1}\big)^*(\Pi)(x_{\sigma(1)},\dots,x_{\sigma(l+1)})(i_{y}^{\alpha}\Xi)(x_{\sigma(l+2)},\dots,x_{\sigma(q+l+1)})\alpha(x) \\
=\sum_{\sigma\in(q+1,l)\text{-unshuf\/f\/les}}\operatorname{sgn}(\sigma) \\
\quad{}\times
\Xi\big(\alpha^{-1}x_{\sigma(1)},\dots,\alpha^{-1}x_{\sigma(q+1)}\big)
\Pi\big(\alpha^{-1}x,\alpha^{-1}x_{\sigma(q+2)},\dots,\alpha^{-1}x_{\sigma(q+l+1)}\big)\alpha(y) \\
\quad{}
-(-1)^{ql}\sum_{\sigma\in(l+1,q)\text{-unshuf\/f\/les}}\operatorname{sgn}(\sigma) \\
\quad{}\times
\Pi\big(\alpha^{-1}x_{\sigma(1)},\dots,\alpha^{-1}x_{\sigma(l+1)}\big)
\Xi\big(\alpha^{-1}y,\alpha^{-1}x_{\sigma(l+2)},\dots,\alpha^{-1}x_{\sigma(q+l+1)}\big)\alpha(x)
\\
=-(-1)^{ql} \bigg( \sum_{\sigma\in(l+1,q)\text{-unshuf\/f\/les}}\operatorname{sgn}(\sigma) \\
\quad{}\times
\alpha(\Xi\otimes x)\big(\alpha^{-1}(\Pi\otimes y)\big(\alpha^{-1}x_{\sigma(1)},\dots,\alpha^{-1}x_{\sigma(l+1)}\big),\alpha^{-1}x_{\sigma(l+2)},\dots,\alpha^{-1}x_{\sigma(q+l+1)}\big) \\
\quad{}-(-1)^{ql}\sum_{\sigma\in(q+1,l)\text{-unshuf\/f\/les}}\operatorname{sgn}(\sigma) \\
\quad{}\times
\alpha(\Pi\otimes y)\big(\alpha^{-1}(\Xi\otimes x)\big(\alpha^{-1}x_{\sigma(1)},\dots,\alpha^{-1}x_{\sigma(q+1)}\big),\alpha^{-1}x_{\sigma(q+2)},\dots,\alpha^{-1}x_{\sigma(q+l+1)}\big)\bigg) \\
=-(-1)^{ql}[\Xi\otimes x,\Pi\otimes y]_\alpha(x_{1},\dots,x_{q+l+1}),\nonumber
\end{gather*}
which implies that \eqref{relation} holds.
\end{proof}

\section{Hom-Lie bialgebras}\label{section5}

The big bracket is a very useful tool to study bialgebra structures. In this section, we follow the classical approach to def\/ine a hom-Lie bialgebra using the hom-big bracket, which turns out to be the same as the one given in~\cite{Yao3}. Furthermore, using the hom-big bracket, it is very easy to give the notions of hom-Lie quasi-bialgebras and hom-quasi-Lie bialgebras.

First we describe hom-Lie algebras and hom-Lie coalgebras using the hom-big bracket.

\begin{pro}\label{mumu}
Let $V$ be a vector space and $\mu \colon \wedge^{2}V\rightarrow V$ a skew-symmetric bilinear map satisfying $\Ad_{\alpha}\mu=\mu$. Then we have
\begin{gather}\label{mu}
\mu(x,y)=-\big\{\{\mu,\alpha^{-1}(x)\big\}_{\alpha},y\}_{\alpha},\qquad \forall \, x, y\in V.
\end{gather}
Furthermore, $\{\mu,\mu\}_{\alpha}=0$ if and only if $\mu$ satisfies the hom-Jacobi identity with respect to~$\alpha$.
\end{pro}

\begin{proof}
By \eqref{hom-NR brackethalf}, \eqref{graded hom-Lie alg1}, \eqref{relation} and the fact that $\Ad_{\alpha}\mu=\mu$, we have
\begin{gather*}
\mu(x,y) = \alpha\mu\big(\alpha^{-1}(x),\alpha^{-1}(y)\big)
=(\mu\circ x)(y)
=[\mu,x]_{\alpha}(y)
=\big(\Ad_{\alpha}\big[\mu,\alpha^{-1}(x)\big]_{\alpha}\big)(y) \\
\hphantom{\mu(x,y)}{}
 = \alpha\big(\big[\mu,\alpha^{-1}(x)\big]_{\alpha}\big)\big(\alpha^{-1}(y)\big)
=\big[\mu,\alpha^{-1}(x)\big]_{\alpha}\circ y
=\big[\big[\mu,\alpha^{-1}(x)\big]_{\alpha},y\big]_{\alpha} \\
\hphantom{\mu(x,y)}{}
=\big[\big\{\mu,\alpha^{-1}(x)\big\}_{\alpha},y\big]_{\alpha}
=-\big\{\big\{\mu,\alpha^{-1}(x)\big\}_{\alpha},y\big\}_{\alpha}.
\end{gather*}

The second conclusion can be obtained directly by Lemma~\ref{lem:mu} and Proposition~\ref{pro:consistant}. Here we give another proof using the hom-big bracket.
By \eqref{G2} and \eqref{mu}, we have
 \begin{gather*}
\big\{\big\{\big\{\big\{\mu,\mu\}_{\alpha},\alpha^{-1}(x)\big\}_\alpha,y\big\}_{\alpha},\alpha(z)\big\}_{\alpha}
=2\big\{\big\{\big\{\mu,\big\{\mu,\alpha^{-2}(x)\big\}_{\alpha}\big\}_{\alpha},y\big\}_{\alpha},\alpha(z)\big\}_{\alpha} \\
=2\big\{\big\{\mu,\big\{\{\mu,\alpha^{-2}(x)\big\}_{\alpha},\alpha^{-1}(y)\big\}_{\alpha}\big\}_{\alpha},\alpha(z)\big\}_{\alpha}
+2\big\{\big\{\big\{\mu,\alpha^{-1}(y)\big\}_{\alpha},\big\{\mu,\alpha^{-1}(x)\big\}_{\alpha}\big\}_{\alpha},\alpha(z)\big\}_{\alpha}\\
=-2\big\{\{\mu,z\}_{\alpha},\big\{\big\{\mu,\alpha^{-1}(x)\big\}_{\alpha},y\big\}_{\alpha}\big\}_{\alpha}
+2\big\{\{\mu,y\}_{\alpha},\big\{\big\{\mu,\alpha^{-1}(x)\big\}_{\alpha},z\big\}_{\alpha}\big\}_{\alpha} \\
\quad{}
+2\big\{\big\{\big\{\mu,\alpha^{-1}y\big\}_{\alpha},z\big\}_{\alpha},\{\mu,x\}_{\alpha}\big\}_{\alpha} \\
=-2(\mu(\alpha(z),\mu(x,y))+\mu(\alpha(y),\mu(z,x))+\mu(\alpha(x),\mu(y,z))),
\end{gather*}
which implies that $\{\mu,\mu\}_{\alpha}=0$ if and only if $\mu$ satisf\/ies the hom-Jacobi identity.
\end{proof}

\begin{pro}\label{deltadelta} Let $V$ be a vector space and $\Delta \colon V\rightarrow\wedge^{2}V$ a linear map satisfying $\Ad_{\alpha}\Delta=\Delta$.
Then we have
 \begin{gather*}
\Delta(x) = \{\Delta, x\}_{\alpha},\qquad \forall \, x\in V,\\
\Delta^{*}(\xi,\eta) = -\big\{\big\{\Delta,\big(\alpha^{3}\big)^{*}(\xi)\big\}_{\alpha},\big(\alpha^{2}\big)^{*}(\eta)\big\}_{\alpha},
\qquad \forall \, \xi, \eta\in V^{*}.
\end{gather*}
Furthermore, $\{\Delta,\Delta\}_{\alpha}=0$ if and only if $\Delta^{*}$ satisfies the hom-Jacobi identity with respect to $\alpha^*$.
\end{pro}

\begin{proof}
 Let $\{e_1,\dots,e_n\}$ be a basis of $V$ and $\{e^1,\dots,e^n\}$ the dual basis of $V^*$. Assume $\Delta=\sum\limits_{i,j,k}\Delta_{k}^{i,j}e^{k}\otimes e_{i}\wedge e_{j}$, then we have
\begin{gather*}
\{\Delta, x\}_{\alpha} = \bigg\{\sum_{i,j,k}\Delta_{k}^{i,j} e^{k}\otimes e_{i}\wedge e_{j},x\bigg\}_{\alpha}=\sum_{i,j,k}\Delta_{k}^{i,j}e^{k}\big(\alpha^{-1}x\big)\alpha(e_{i}\wedge e_{j}) \\
\hphantom{\{\Delta, x\}_{\alpha}}{}
 = (\Ad_\alpha\Delta)(x)=\Delta(x).\nonumber
\end{gather*}
Furthermore, by $\Ad_{\alpha}\Delta=\Delta$, we have
\begin{gather*}
\Delta^{*}(\xi,\eta) = -\Delta^{*}(\eta,\xi)=-(\Ad_{\alpha}\Ad_{\alpha}\Delta)(\eta,\xi) \\
\hphantom{\Delta^{*}(\xi,\eta)}{}
 = -\bigg( \sum_{i,j,k}\Delta_{k}^{i,j}\big(\alpha^{-2}\big)^{*}\big(e^{k}\big)\otimes\alpha^{2}(e_{i})\wedge\alpha^{2}(e_{j})\bigg) (\eta,\xi) \\
 \hphantom{\Delta^{*}(\xi,\eta)}{}
 = -\sum_{i,j,k}\Delta_{k}^{i,j}\big(\alpha^{-2}\big)^{*}\big(e^{k}\big)\big(\eta\big(\alpha^{2}e_{i}\big)\xi\big(\alpha^{2}e_{j}\big)
 -\xi\big(\alpha^{2}e_i\big)\eta\big(\alpha^{2}e_{j}\big)\big) \\
 \hphantom{\Delta^{*}(\xi,\eta)}{}
 = -\bigg\{\sum_{i,j,k}\Delta_{k}^{i,j} \big(\alpha^{-1}\big)^{*}\big(e^{k}\big)\otimes\big(\xi\big(\alpha^{2}e_{j}\big)\alpha(e_{i})-\xi\big(\alpha^{2}e_{i}\big)\alpha(e_{j})\big), \big(\alpha^{2}\big)^{*}(\eta)\bigg\}_{\alpha} \\
 \hphantom{\Delta^{*}(\xi,\eta)}{}
 = -\big\{\big\{\Delta,\big(\alpha^{3}\big)^{*}(\xi)\big\}_{\alpha},\big(\alpha^{2}\big)^{*}(\eta)\big\}_{\alpha}.
\end{gather*}
Finally, we have
\begin{gather*}
 \big\{\big\{\big\{\{\Delta,\Delta\}_{\alpha},\big(\alpha^{5}\big)^{*}(\xi)\big\}_{\alpha},
 \big(\alpha^{4}\big)^{*}(\eta)\big\}_{\alpha},\big(\alpha^{3}\big)^{*}(\delta)\big\}_{\alpha} \\
 \qquad{}= 2\big\{\big\{\big\{\Delta,\big\{\Delta,\big(\alpha^{6}\big)^{*}(\xi)\big\}_{\alpha}\big\}_{\alpha},
 \big(\alpha^{4}\big)^{*}(\eta)\big\}_{\alpha},\big(\alpha^{3}\big)^{*}(\delta)\big\}_{\alpha} \\
 \qquad{}
 = 2\big\{\big\{\Delta,\big\{\big\{\Delta,\big(\alpha^{6}\big)^{*}(\xi)\big\}_{\alpha},
 \big(\alpha^{5}\big)^{*}(\eta)\big\}_{\alpha}\big\}_{\alpha},\big(\alpha^{3}\big)^{*}(\delta)\big\}_{\alpha}\\
 \qquad\quad{}
 +2\big\{\big\{\big\{\Delta,\big(\alpha^{5}\big)^{*}(\eta)\big\}_{\alpha},\big\{\Delta,\big(\alpha^{5}\big)^{*}(\xi)\big\}_{\alpha}
 \big\}_{\alpha},\big(\alpha^{3}\big)^{*}(\delta)\big\}_{\alpha} \\
 \qquad{}
 = -2\big\{\big\{\Delta,\big(\alpha^{4}\big)^{*}(\delta)\big\}_{\alpha},\big\{\big\{\Delta,
 \big(\alpha^{5}\big)^{*}(\xi)\big\}_{\alpha},\big(\alpha^{4}\big)^{*}(\eta)\big\}_{\alpha}\big\}_{\alpha} \\
 \qquad\quad{}
 +2\big\{\big\{\Delta,\big(\alpha^{4}\big)^{*}(\eta)\big\}_{\alpha},
 \big\{\big\{\Delta,\big(\alpha^{5}\big)^{*}(\xi)\big\}_{\alpha},\big(\alpha^{4}\big)^{*}(\delta)\big\}_{\alpha}\big\}_{\alpha}\\
 \qquad\quad {}
 +2\big\{\big\{\big\{\Delta,\big(\alpha^{5}\big)^{*}(\eta)\big\}_{\alpha},
 \big(\alpha^{4}\big)^{*}(\delta)\big\}_{\alpha},\big\{\Delta,\big(\alpha^{4}\big)^{*}(\xi)\big\}_{\alpha}\big\}_{\alpha} \\
 \qquad{}
 = 2\big\{\big\{\Delta,\big(\alpha^{4}\big)^{*}(\delta)\big\}_{\alpha},\Delta^{*}\big(\big(\alpha^{2}\big)^{*}(\xi),
 \big(\alpha^{2}\big)^{*}(\eta))\big\}_{\alpha} \\
\qquad\quad{} -2\big\{\big\{\Delta,\big(\alpha^{4}\big)^{*}(\eta)\big\}_{\alpha},\Delta^{*}\big(\big(\alpha^{2}\big)^{*}(\xi),
\big(\alpha^{2}\big)^{*}(\delta))\big\}_{\alpha}\\
\qquad\quad{}
 -2\big\{\Delta^{*}\big(\big(\alpha^{2}\big)^{*}(\eta),\big(\alpha^{2}\big)^{*}(\delta)),\big\{\Delta,\big(\alpha^{4}\big)^{*}(\xi)\big\}_{\alpha}\big\}_{\alpha} \\
 \qquad{}
 = -2\big(\Delta^{*}\big(\alpha^{*}(\delta),\Delta^{*}(\xi,\eta)\big)+\Delta^{*}\big(\alpha^{*}(\eta),\Delta^{*}(\delta,\xi)\big)
+\Delta^{*}\big(\alpha^{*}(\xi),\Delta^{*}(\eta,\delta)\big)\big),
\end{gather*}
which implies that $\{\Delta,\Delta\}_{\alpha}=0$ if and only if $\Delta^{*}$ satisf\/ies the hom-Jacobi identity with respect to~$\alpha^*$.
\end{proof}

Now we give the notion of a hom-Lie bialgebra using the hom-big bracket.

\begin{defi}\looseness=-1 Let $V$ be a vector space, $\alpha\in {\rm GL}(V)$, $\mu\colon \wedge^2 V\longrightarrow V$ and $\Delta \colon V\rightarrow\wedge^{2}V$ linear maps satisfying $\Ad_{\alpha}\mu=\mu$ and $\Ad_{\alpha}\Delta=\Delta$.
The quadruple
$(V,\mu,\Delta,\alpha)$ is a hom-Lie bialgebra if
\begin{gather}\label{eq:bialgebra}
\{\mu+\Delta,\mu+\Delta\}_{\alpha}=0.
\end{gather}
\end{defi}

It is helpful to describe a hom-Lie bialgebra using the usual algebraic language.
\begin{pro}
With the above notations, a quadruple $(V, \mu, \Delta, \alpha)$ is a hom-Lie bialgebra if and only if
\begin{itemize}\itemsep=0pt

\item[$(i)$] $(V, \mu, \alpha)$ is a hom-Lie algebra;

\item[$(ii)$] $(V^{*}, \Delta^{*}, \alpha^{*})$ is a hom-Lie algebra;

\item[$(iii)$] $\Delta(\mu(x, y))=\ad^\mu_{\alpha(x)}\Delta(y)-\ad^\mu_{\alpha(y)}\Delta(x)$, for all $ x,y\in V$, where $\ad^\mu$ is the action of the hom-Lie algebra $(V, \mu, \alpha)$ on $\wedge^2V $ given by $\ad^\mu_x(y\wedge z)=\mu(x,y)\wedge z+y\wedge\mu(x,z)$.
\end{itemize}
\end{pro}

\begin{proof}
 It is obvious that \eqref{eq:bialgebra} holds if and only if
\begin{gather*}
 \{\mu,\mu\}_{\alpha}=0,\qquad \{\Delta,\Delta\}_{\alpha}=0,\qquad \{\mu,\Delta\}_{\alpha}=0.
\end{gather*}
 By Propositions~\ref{mumu} and~\ref{deltadelta}, (i) and (ii) hold obviously. One can also prove that $\{\mu,\Delta\}_{\alpha}=0$ if and only if~(iii) holds. We omit details.
 \end{proof}

\begin{rmk}
The def\/inition of a hom-Lie bialgebra given above is the same as the one given in~\cite{Yao3}. However, to obtain the Manin triple theory, we need to follow the approach given in~\cite{sheng1}.
\end{rmk}

At the end of this section, we give the notions of a hom-Lie quasi-bialgebra and a hom-quasi-Lie bialgebra.
\begin{defi}Let $V$ be a vector space, $\alpha\in {\rm GL}(V)$, $\mu\colon \wedge^2 V\longrightarrow V$ and $\Delta \colon V\rightarrow\wedge^{2}V$ linear maps satisfying $\Ad_{\alpha}\mu=\mu$ and $\Ad_{\alpha}\Delta=\Delta$.
 \begin{itemize}\itemsep=0pt
 \item[(i)]
 The $5$-tuple $(V,\mu,\Delta,\alpha,\phi)$ is called a hom-Lie quasi-bialgebra, if
 \begin{gather}\label{quasi-lie}
 \{\phi+\mu+\Delta,\phi+\mu+\Delta\}_{\alpha}=0
 \end{gather}
 with $\phi\in\wedge^{3}V$ satisfying $ {\alpha}(\phi)=\phi$.

 \item[(ii)]
 The $5$-tuple $(V,\mu,\Delta,\psi,\alpha)$ is called a hom-quasi-Lie bialgebra, if
 \begin{gather*}
 \{\mu+\Delta+\psi,\mu+\Delta+\psi\}_{\alpha}=0
 \end{gather*}
 with $\psi\in\wedge^{3}V^{*}$ satisfying ${(\alpha^{-1})}^*(\psi)=\psi$.
 \end{itemize}
\end{defi}

Using the usual algebraic language, we have

\begin{pro}
With the above notations, a $5$-tuple $(V, \mu, \Delta, \alpha,\phi)$ is a hom-Lie quasi-bialgebra if and only if
\begin{itemize}\itemsep=0pt
\item[$(i)$]
$(V, \mu, \alpha)$ is a hom-Lie algebra;

\item[$(ii)$]
$ \Delta^{*}(\alpha^{*}\xi,\Delta^*(\eta,\delta))+{\rm c.p.}=\ad_{\phi(\cdot,(\alpha^2)^*\eta,(\alpha^2)^*\delta)}^*\alpha^*\xi+{\rm c.p.}$, for all $\xi,\eta,\gamma\in V^*$;

\item[$(iii)$]
$\Delta(\mu(x, y))=\ad^\mu_{\alpha(x)}\Delta(y)-\ad^\mu_{\alpha(y)}\Delta(x)$, for all $x,y\in V$;

\item[$(iv)$] $\dM_*\phi=0$,
where $\dM_*$ is given by \eqref{eq:newd} determined by $\Delta^*$.
\end{itemize}
\end{pro}

\begin{proof}
 By \eqref{quasi-lie}, we have
\begin{gather*}
\{\mu,\mu\}_\alpha=0,\qquad \tfrac{1}{2}\{\Delta,\Delta\}_\alpha+\{\mu,\phi\}_\alpha=0,\qquad \{\mu,\Delta\}_\alpha=0,\qquad \{\Delta,\phi\}_\alpha=0,
\end{gather*}
which gives (i)--(iv) respectively.
\end{proof}

Similarly, we have
\begin{pro}
With the above notations, a $5$-tuple $(V, \mu, \Delta, \alpha,\psi)$ is a hom-quasi-Lie bialgebra if and only if
\begin{itemize}\itemsep=0pt
\item[$(i)$]
$\mu(\alpha(x),\mu(y,z))+{\rm c.p.}=\ad_{\psi(\cdot,\alpha^{2}x,\alpha^{2}y)}^*\alpha (z)+{\rm c.p.}$, for all $x,y,z\in V$;

\item[$(ii)$]
$(V^*, \Delta^*, \alpha^*)$ is a hom-Lie algebra;

\item[$(iii)$]
$\Delta(\mu(x, y))=\ad^\mu_{\alpha(x)}\Delta(y)-\ad^\mu_{\alpha(y)}\Delta(x)$, for all $x,y\in V$;

\item[$(iv)$] $\dM\psi=0$,
where $\dM$ is given by~\eqref{eq:newd} determined by~$\mu$.
\end{itemize}
\end{pro}

\section[Hom-Nijenhuis operators and hom-$\huaO$-operators]{Hom-Nijenhuis operators and hom-$\boldsymbol{\huaO}$-operators}\label{section6}

In this section, we give the notion of a hom-Nijenhuis operator using the hom-big bracket. We show that a hom-Nijenhuis operator gives rise to a trivial deformation. Furthermore, a new def\/inition of a Hom-$\huaO$-operator is given.

\begin{defi}
Let $(V,\mu,\alpha)$ be a regular hom-Lie algebra. A~linear map $N\colon V\longrightarrow V$ satisfying $\Ad_\alpha N=N$ is called a hom-Nijenhuis operator if
\begin{gather*}
\{N,\{N,\mu\}_{\alpha}\}_{\alpha}-\{N\circ N,\mu\}_{\alpha}=0,
\end{gather*}
where $\circ$ is def\/ined by~\eqref{hom-NR brackethalf}.
\end{defi}

The next proposition characterizes a hom-Nijenhuis operator using the usual algebraic formula. To be simple, we write $\mu(x,y)$ by $[x,y]$ in the sequel.

\begin{pro}\label{hom-Nijenhuis operator}
Let $(V,\mu,\alpha)$ be a regular hom-Lie algebra. A linear map $N\colon V\longrightarrow V$ satisfying $\Ad_\alpha N=N$ is a hom-Nijenhuis operator if and only if
\begin{gather}\label{Nijenhuis}
[Nx,Ny]=N\big[N\alpha^{-1}x,y\big]+N\big[x,N\alpha^{-1}y\big]-N^{2}\big[\alpha^{-1}x,\alpha^{-1}y\big],\qquad\forall \, x,y\in V.
\end{gather}
\end{pro}

\begin{proof}
By the fact $\Ad_\alpha N=N$, we have
\begin{gather}\label{N}
-\{N,y\}_{\alpha}=[N,y]_{\alpha}=N\circ y=(\Ad_{\alpha}N)(y)=N(y).
\end{gather}
Therefore, we have
\begin{gather}
-\big\{\big\{\{N,\mu\}_{\alpha},\alpha^{-1}x\big\}_{\alpha},y\big\}_{\alpha} = -\big\{\{N,\big\{\mu,\alpha^{-2}x\big\}_{\alpha}\big\}_{\alpha},y\}_{\alpha}
+\big\{\big\{\big\{N,\alpha^{-2}x\big\}_{\alpha},\mu\big\}_{\alpha},y\big\}_{\alpha}\nonumber\\
\qquad{}
 = -\big\{N,\big\{\{\mu,\alpha^{-2}x\big\}_{\alpha},\alpha^{-1}y\big\}_{\alpha}\}_{\alpha}-
 \big\{\big\{N,\alpha^{-1}y\big\}_{\alpha},\big\{\mu,\alpha^{-1}x\big\}_{\alpha}\big\}_{\alpha}\nonumber\\
\qquad\quad{} +\big\{\big\{N,\alpha^{-1}x\big\}_{\alpha},\big\{\mu,\alpha^{-1}y\big\}_{\alpha}\big\}_{\alpha}-
\big\{\big\{\big\{N,\alpha^{-2}x\big\}_{\alpha},\alpha^{-1}y\big\}_{\alpha},\mu\big\}_{\alpha}\nonumber\\
\qquad{} = -N\big[\alpha^{-1}x,\alpha^{-1}y\big]+\big[x,N\alpha^{-1}y\big]+\big[N\alpha^{-1}x,y\big].\label{n1}
\end{gather}
By \eqref{N} and \eqref{n1}, we have
\begin{gather}\label{N1}
\big\{N, \big\{\big\{\{N,\mu\}_{\alpha},\alpha^{-1}x\big\}_{\alpha},y\big\}_{\alpha}\big\}_{\alpha}
=N\big[N\alpha^{-1}x,y\big]+N\big[x,N\alpha^{-1}y\big]-N^{2}\big[\alpha^{-1}x,\alpha^{-1}y\big].
\end{gather}
By \eqref{mu}, \eqref{N}--\eqref{N1}, we have
\begin{gather*}
 \{\{\{N,\{N,\mu\}_{\alpha}\}_{\alpha}
-\{N\circ N,\mu\}_{\alpha},x\}_{\alpha}, \alpha y\}_{\alpha}
 = \{\{\{N,\{N,\mu\}_{\alpha}\}_{\alpha},x\}_{\alpha}, \alpha y\}_{\alpha} \\
 \qquad\quad{}
 -\big\{\big\{N\circ N,\big\{\mu,\alpha^{-1}x\big\}_{\alpha}\big\}_{\alpha},\alpha y\big\}_{\alpha}
+\big\{\big\{\big\{N\circ N,\alpha^{-1}x\big\}_{\alpha},\mu\big\}_{\alpha},\alpha y\big\}_{\alpha} \\
\qquad{}
 = \big\{\big\{N,\big\{\{N,\mu\}_{\alpha},\alpha^{-1}x\big\}_{\alpha}\big\}_{\alpha}, \alpha y\big\}_{\alpha}
-\big\{\big\{\big\{N,\alpha^{-1}x\big\}_{\alpha},\{N,\mu\}_{\alpha}\big\}_{\alpha},\alpha y\big\}_{\alpha} \\
\qquad\quad{}
 -\big\{N\circ N,\big\{\big\{\mu,\alpha^{-1}x\big\}_{\alpha},y\big\}_{\alpha}\big\}_{\alpha}
-\{\{N\circ N,y\}_{\alpha},\{\mu,x\}_{\alpha}\}_{\alpha}+[N^{2}(\alpha^{-1}x),\alpha y] \\
\qquad {}
 = \big\{N,\big\{\big\{\{N,\mu\}_{\alpha},\alpha^{-1}x\big\}_{\alpha},y\big\}_{\alpha}\big\}_{\alpha}
 +\{\{N,y\}_{\alpha},\{\{N,\mu\}_{\alpha},x\}_{\alpha}\}_{\alpha} \\
 \qquad\quad{}
 +\big\{\big\{\{N,\mu\}_{\alpha},N\big(\alpha^{-1}x\big)\big\}_{\alpha},\alpha y\big\}_{\alpha}
 -N^{2}\big[\alpha^{-1}x,\alpha^{-1}y\big]\\
 \qquad\quad{}
 +\big[\alpha x,N^{2}\alpha^{-1}y\big]+\big[N^{2}\alpha^{-1}x,\alpha y\big] \\
\qquad{} = \big({-}N^{2}\big[\alpha^{-1}x,\alpha^{-1}y\big]+N\big[x,N\alpha^{-1}y\big]+N\big[N\alpha^{-1}x,y\big]\big) \\
 \qquad\quad{}
 +\big(N\big[x,N\alpha^{-1}y\big]-\big[\alpha x,N^{2}\alpha^{-1}y\big]-[Nx,Ny]) \\
 \qquad\quad{}
 +\big(N\big[N\alpha^{-1}x,y\big]-[Nx,Ny]-\big[N^{2}\alpha^{-1}x,\alpha y\big]\big) \\
 \qquad\quad{}
 -N^{2}\big[\alpha^{-1}x,\alpha^{-1}y\big]+\big[\alpha x,N^{2}\alpha^{-1}y\big]+
 \big[N^{2}\alpha^{-1}x,\alpha y\big] \\
 \qquad{}
 = 2\big(N\big[N\alpha^{-1}x,y\big]+N\big[x,N\alpha^{-1}y\big]-N^{2}\big[\alpha^{-1}x,\alpha^{-1}y\big]-[Nx,Ny]).
\end{gather*}
Therefore, $N$ is a hom-Nijenhuis operator if and only if~\eqref{Nijenhuis} holds.
\end{proof}

\begin{rmk}\label{rmk:Nijenhuis}
The def\/inition of a hom-Nijenhuis operator given above is dif\/ferent from the one given in~\cite{sheng3}. In~\cite{sheng3}, a hom-Nijenhuis operator on a~hom-Lie algebra $(V,[\cdot,\cdot],\alpha)$ is def\/ined to be a linear map $N\colon V\longrightarrow V$ satisfying $\alpha\circ N=N\circ \alpha$ and the following integrability condition
\begin{gather}\label{eq:Nijenhuisold}
[Nx,Ny]=N[Nx,y]+N[x,Ny]-N^{2}[ x, y].
\end{gather}

Comparing with \eqref{Nijenhuis}, \eqref{eq:Nijenhuisold} does not contain the information about the homomorphism~$\alpha$. Thus, We believe that the current def\/inition is more reasonable. This justif\/ies the usage of the hom-big bracket.
\end{rmk}

Now we consider deformations of a hom-Lie algebra.
Let $(V,[\cdot,\cdot],\alpha)$ be a hom-Lie algebra and $\omega\in \Hom(\wedge^{2}V,V)$, def\/ine
\begin{gather*}
[x,y]_{t}:=[x,y]+t\omega(x,y),\qquad t\in\mathbb R.
\end{gather*}
For all $t$, $(V,[\cdot,\cdot]_{t},\alpha)$ is a hom-Lie algebra if and only if
\begin{gather}
\Ad_\alpha\omega = \omega,\nonumber\\
\label{wjacobi}
[\omega(x,y),\alpha(z)]+\omega([x,y],\alpha(z))+{\rm c.p.}(x,y,z) = 0,\\
\label{eq:dw0}
\omega(\omega(x,y),\alpha(z))+{\rm c.p.}(x,y,z) = 0.
\end{gather}
Note that \eqref{wjacobi} is equivalent to $(\dM\omega)(\alpha x,\alpha y,\alpha z) =0$ and~\eqref{eq:dw0} is equivalent to $[\omega,\omega]_{\alpha}(\alpha x,\alpha y,\alpha z)$ $=0$. That is to say, $(V,[\cdot,\cdot]_{t},\alpha)$ is a hom-Lie algebra for all $t$ if and only if $(V,\omega,\alpha)$ is a hom-Lie algebra and $\dM \omega=0$. In this case, we say that $\omega$ generates a 1-parameter inf\/initesimal deformation.

\begin{defi}
A deformation is said to be trivial, if there exists a linear operator $N\colon V\longrightarrow V$ satisfying $\Ad_\alpha N=N$ such that
 \begin{gather}
 \label{trivial deformation}
 (\alpha+tN)[x,y]_{t}=[(\alpha+tN)(x),(\alpha+tN)(y)].
 \end{gather}
\end{defi}
The condition \eqref{trivial deformation} is equivalent to
\begin{gather*}
\omega(x,y) = \big[x,N\alpha^{-1}y\big]+\big[N\alpha^{-1}x,y\big]-N\big[\alpha^{-1}x,\alpha^{-1}y\big],\\
N(\omega(x,y)) = [Nx,Ny].
\end{gather*}
Therefore, $N$ is a hom-Nijenhuis operator. Thus, a trivial deformation gives rise to a hom-Nijenhuis operator. The converse is also true.

\begin{thm}
Let $N$ be a hom-Nijenhuis operator. Then a deformation can be obtained by putting
\begin{gather*}
\omega(x,y)=-\big\{\big\{\{N,\mu\}_{\alpha},\alpha^{-1}x\big\}_\alpha,y\big\}_{\alpha}.
\end{gather*}
Furthermore, this deformation is trivial.
\end{thm}

\begin{proof}
 Obviously, $\omega=\dM N$. Therefore, \eqref{wjacobi} holds naturally. By $\Ad_\alpha N=N$ and $\Ad_\alpha\mu=\mu$, we can deduce that $\Ad_\alpha\omega=\omega$. Finally, we need to check the hom-Jacobi identity for $\omega$, which follows from the Nijenhuis condition~\eqref{Nijenhuis}. We omit details. Therefore, $\omega$ generates a trivial deformation.
Furthermore, also by~\eqref{Nijenhuis}, it is straightforward to see that~\eqref{trivial deformation} holds. Thus, this deformation is trivial.
\end{proof}

As in the classical case, any polynomial of a Nijenhuis operator is still a Nijenhuis operator. The following formula can be obtained by straightforward computations.

\begin{lem}\label{nijenhuisthm}
Let $N$ be a hom-Nijenhuis operator acting on a hom-Lie algebra $(V,[\cdot,\cdot],\alpha)$. Then for all $ i, j\in\mathbb{N}$, there holds
\begin{gather*}
\big[N^{i}x,N^{j}y\big]-N^{i}\big[x,N^j\alpha^{-i}y\big]-N^{j}\big[N^{i}\alpha^{-j}x,y\big]
+N^{i+j}\big[\alpha^{-j}x,\alpha^{-i}y\big]=0,\qquad \forall\, x,y\in V.
\end{gather*}
\end{lem}

\begin{thm}
Let $N$ be a hom-Nijenhuis operator acting on a hom-Lie algebra $(V,[\cdot,\cdot],\alpha)$. Then for any polynomial $P(z)=\sum\limits_{i=0}^{n}c_{i}z^{i}$,
 $P(N):=\sum\limits_{i=0}^{n}c_{i}N\circ \stackrel{i}{\cdots}\circ N $ is a Nijenhuis operator, where $\circ $ is defined by~\eqref{hom-NR brackethalf}.
\end{thm}

\begin{proof}
By Lemma~\ref{nijenhuisthm}, we have
\begin{gather*}
 [P(N)x,P(N)y]-P(N)\big[P(N)\alpha^{-1}x,y\big]-P(N)\big[x,P(N)\alpha^{-1}y\big]+(P(N))^2\big[\alpha^{-1}x,\alpha^{-1}y\big] \\
 \qquad{}
 = \sum_{i,j=0}^{n}c_{i}c_{j}\big(\big[N^{i}\alpha^{-i+1}x,N^{j}\alpha^{-j+1}y\big]-N^{j}\alpha^{-j+1}\big[N^{i}\alpha^{-i}x,y\big] \\
\qquad\quad{} -N^{i}\alpha^{-i+1}\big[x,N^{j}\alpha^{-j}y\big]+N^{i+j}\alpha^{-i-j+2}\big[\alpha^{-1}x,\alpha^{-1}y\big]\big) \\
\qquad{}
 = \sum_{i,j=0}^{n}c_{i}c_{j}\big(\big[N^{i}\alpha^{-i+1}x,N^{j}\alpha^{-j+1}y\big]-N^{j}\big[N^{i}\alpha^{-i-j+1}x,\alpha^{-j+1}y\big] \\
\qquad\quad{}
 -N^{i}\big[\alpha^{-i+1}x,N^{j}\alpha^{-i-j+1}y\big]+N^{i+j}\big[\alpha^{-i-j+1}x,\alpha^{-i-j+1}y\big]\big)=0.
\end{gather*}
Therefore, $P(N)$ is a Nijenhuis operator.
\end{proof}

At the end of this section, we introduce a new def\/inition of a hom-$\huaO$-operator, which is a~generalization of an $\huaO$-operator introduced by Kupershmidt in~\cite{Kupershmidt2}.
\begin{defi}
Let $(V,[\cdot,\cdot],\alpha)$ be a hom-Lie algebra and $\rho\colon V\longrightarrow\gl(W)$ a representation of $(V,[\cdot,\cdot],\alpha)$ on~$W$ with respect to $\beta\in {\rm GL}(W)$. A linear map $T\colon W\rightarrow V$ is called a hom-$\huaO$-operator if $T$ satisf\/ies
 \begin{gather*}
 T\circ \beta = \alpha\circ T,\\
 [Tu,Tv]_{\frkg} = T\big(\rho\big(T\big(\beta^{-1}u\big)\big)v-\rho\big(T\big(\beta^{-1}v\big)\big)u\big).
 \end{gather*}
 \end{defi}

\begin{lem}
With the above notations, a linear map $T\colon W\rightarrow V$ is a hom-$\huaO$-operator if and only if
$\left(\begin{smallmatrix}0&T\\0&0\end{smallmatrix}\right)$ is a hom-Nijenhuis operator for the semidirect product hom-Lie algebra $V\ltimes_{\rho}W$.
\end{lem}
\begin{proof}
 By straightforward computations.
 \end{proof}

\begin{rmk}\label{rmk:O-operator}
 As in the case of hom-Nijenhuis operators, the above def\/inition of a hom-$\huaO$-operator is dif\/ferent from the one given in~\cite{sheng1}. Now our principle is that~$T$ is a hom-$\huaO$-operator if and only if $ \left(\begin{smallmatrix}0&T\\0&0\end{smallmatrix}\right)$ is a hom-Nijenhuis operator as the above lemma shows. Since the present def\/inition of a hom-Nijenhuis operator is dif\/ferent from the one given in \cite{sheng3} (see Remark~\ref{rmk:Nijenhuis}), it is reasonable that the def\/inition of a hom-$\huaO$-operator is also dif\/ferent from the old one. Recently, some applications of hom-$\huaO$-operators were given in~\cite{Zhang}.
\end{rmk}

 As in the classical case, a hom-$\huaO$-operator can give rise to a hom-right-symmetric algebra.
\begin{pro}
Let $(V,[\cdot,\cdot],\alpha)$ be a hom-Lie algebra and $\rho\colon V\longrightarrow\gl(W)$ a representation of $(V,[\cdot,\cdot],\alpha)$ on $W$ with respect to $\beta\in {\rm GL}(W)$.
Suppose that $T\colon W\rightarrow V$ is a hom-$\huaO$-operator. Then $(W,*,\beta)$ is a hom-right-symmetric algebra, where the multiplication $*$ is given by
\begin{gather*}
u*v=\rho\big(T\big(\beta^{-1}v\big)\big)(u),\qquad\forall \,u,v\in W.
\end{gather*}
\end{pro}

\begin{proof}
By $ T\circ \beta=\alpha\circ T$ and the fact that $\rho$ is a representation, we have
\begin{gather*}
\beta(u*v)=\beta\big(\rho\big(T\big(\beta^{-1}v\big)\big)(u)\big)
=\rho\big(\alpha\big(T\big(\beta^{-1}v)\big)\big)( \beta(u))=\rho(T(v))(\beta(u))=\beta(u)*\beta(v),
\end{gather*}
which implies that $\beta$ is an algebra homomorphism. Furthermore, we have
\begin{gather*}
 (u*v)*\beta(w)-\beta(u)*(v*w)-(u*w)*\beta(v)+\beta(u)*(w*v) \\
 = \rho(T(w))\big(\rho\big(T\big(\beta^{-1}v\big)\big)(u))-\rho\big(T\big(\beta^{-1}\rho\big(T\big(\beta^{-1}w\big)\big)(v)\big)\big)(\beta(u)) \\
\quad{} -\rho(T(v))\big(\rho\big(T\big(\beta^{-1}w\big)\big)(u)\big)+\rho\big(T\big(\beta^{-1}\rho\big(T\big(\beta^{-1}v\big)\big)(w)\big)\big)(\beta(u)) \\
 = \rho T\big(\beta^{-1}\rho\big(T\big(\beta^{-1}v\big)\big)(w)-\beta^{-1}\rho\big(T\big(\beta^{-1}w\big)\big)(v)\big)(\beta(u)) \\
 \quad{} -\rho\big(\alpha\circ T\big(\beta^{-1}v\big)\big)\big(\rho\big(T\big(\beta^{-1}w\big)\big)(u)\big)+\rho\big(\alpha\circ T\big(\beta^{-1}w\big)\big)\big(\rho\big(T\big(\beta^{-1}v\big)\big)(u)\big) \\
 = \rho\alpha^{-1}\circ T\big(\rho\big(T\big(\beta^{-1}v\big)\big)(w)-\rho\big(T\big(\beta^{-1}w\big)\big)(v)\big)(\beta(u))
 -\rho\big(\big[T\big(\beta^{-1}v\big),T\big(\beta^{-1}w\big)\big]\big)(\beta(u)) \\
 = \rho\alpha^{-1}([T(v),T(w)])(\beta(u))-\rho\big(\big[T\big(\beta^{-1}v\big),T\big(\beta^{-1}w\big)\big] \big)(\beta(u)) \\
 = \rho\big(\big[\alpha^{-1}T(v),\alpha^{-1}T(w)\big] \big)(\beta(u))-\rho\big(\big[T\big(\beta^{-1}v\big),T\big(\beta^{-1}w\big)\big] \big)(\beta(u)) \\
 = \rho\big(\big[T\circ \beta^{-1}(v),T\circ \beta^{-1}(w)\big] \big)(\beta(u))-\rho\big(\big[T\big(\beta^{-1}v\big),T\big(\beta^{-1}w\big)\big] \big)(\beta(u))
=0.
\end{gather*}
Therefore, $(W,*,\beta)$ is a hom-right-symmetric algebra.
\end{proof}

\subsection*{Acknowledgements}

We give our warmest thanks to the editor and referees for very useful comments that improve the paper.
This research is supported by NSFC (11101179, 11471139) and NSF of Jilin Province (20140520054JH).

\pdfbookmark[1]{References}{ref}
\LastPageEnding

\end{document}